\RequirePackage{amsmath}

\documentclass{llncs}
\usepackage[backgroundcolor=white,linecolor=black,textsize=scriptsize,disable]{todonotes}

\usepackage{amssymb}
\usepackage{xspace}
\usepackage{microtype}
\usepackage{multirow,tabularx}
\usepackage{makecell}
\usepackage{booktabs}

\let\doendproof\endproof
\renewcommand{\endproof}{\hfill$\qed$\doendproof}

\usepackage{url}
\usepackage[inline]{enumitem}

\definecolor{darkgreen}{RGB}{0,100,0}
\definecolor{darkblue}{RGB}{0,0,139}
\definecolor{darkred}{RGB}{139,0,0}
\definecolor{darkmagenta}{RGB}{139,0,139}
\definecolor{darkorange}{RGB}{255,140,0}

\title{Multilevel Planarity}

\author{Lukas Barth
  \and Guido Br\"uckner
  \and Paul Jungeblut
  \and Marcel Radermacher}

\institute{Department of Computer Science, Karlsruhe  Institute of Technology,
	Germany	\email{lukas.barth@kit.edu}, \email{brueckner@kit.edu}, \email{paul.jungeblut@student.kit.edu}, \email{radermacher@kit.edu}
}
\date{}

\begin{document}

\maketitle

\begin{abstract}
  In this paper, we introduce and study the \emph{multilevel-planari\-ty testing} problem, which 
  is a generalization of upward planarity and level planarity.
  Let~$G = (V, E)$ be a directed graph and let~$\ell: V \to \mathcal P(\mathbb Z)$ be a function that assigns a finite set of integers to each vertex.
  A~multilevel-planar drawing of~$G$ is a planar drawing of~$G$ such that the~$y$-coordinate of each vertex~$v \in V$ is~$y(v) \in \ell(v)$, and each edge is drawn as a strictly~$y$-monotone curve.

  We present linear-time algorithms for testing multilevel planarity of embedded graphs with a single source and of oriented cycles.
  Complementing these algorithmic results, we show that multilevel-planarity testing is~\textsf{NP}-complete even in very restricted cases.
\end{abstract}

\section{Introduction}
Testing a given graph for planarity, and, if it is planar, finding a planar embedding, are classic algorithmic problems.
However, one is often not interested in just any planar embedding, but in one that has some additional properties.
Examples of such properties include that a given existing partial drawing should be extended~\cite{Angelini2015,Jelinek2013} or that some parts of the graph should appear clustered together~\cite{DiBattista2008,Jelinkova2009}.

There also exist notions of planarity specifically tailored to directed graphs.
An \emph{upward-planar drawing} is a planar drawing where each edge is drawn as a strictly~$y$-monotone curve.
While testing upward planarity of a graph is an \textsf{NP}-complete problem in general~\cite{Garg2002}, efficient algorithms are known for single-source graphs and for embedded graphs~\cite{Bertolazzi1994,Bertolazzi1998}.
One notable constrained version of upward planarity is that of level planarity.
A \emph{level graph} is a directed graph~$G = (V, E)$ together with a level assignment~$\gamma : V \to \mathbb{Z}$ that assigns an integer level to each vertex and satisfies~$\gamma(u) < \gamma(v)$ for all~$(u, v) \in E$.
A drawing of~$G$ is level planar if it is upward planar, and we have~$y(v) = \gamma(v)$ for the~$y$-coordinate of each vertex~$v \in V$.
Level-planarity testing and embedding is feasible in linear time~\cite{Juenger1999}.
There exist further level-planarity variants on the cylinder and on the torus~\cite{Angelini2016,Bachmaier2005} and there has been considerable research on further-constrained versions of level planarity.
Examples include ordering the vertices on each level according to so-called constraint trees~\cite{Angelini2015-proper,Harrigan2008}, clustered level planarity~\cite{Angelini2015-proper,Forster2004}, partial level planarity~\cite{Brueckner2017} and ordered level planarity~\cite{Klemz2017}.

\paragraph{Contribution and Outline.}

In this paper, we introduce and study the \emph{multilevel-planarity testing} problem.
Let~$\mathcal P(\mathbb Z)$ denote the power set of integers.
The input of the multilevel-planarity testing problem consists of a directed graph~$G = (V, E)$ together with a function~$\ell: V \to \mathcal P(\mathbb Z)$, called a multilevel assignment, which assigns admissible levels, represented as a set of integers, to each vertex.
A~multilevel-planar drawing of~$G$ is a planar drawing of~$G$ such that for the~$y$-coordinate of each vertex~$v \in V$ it holds that~$y(v) \in \ell(v)$, and each edge is drawn as a strictly~$y$-monotone curve.
We start by discussing some preliminaries, including the relationship between multilevel planarity and existing planarity variants in Section~\ref{sec:preliminaries}.
Then, we present linear-time algorithms that test multilevel planarity of embedded single-source graphs and of oriented cycles with multiple sources in Sections~\ref{sec:embeddedSingleSourceGraphs} and~\ref{sec:cycles}, respectively.
In Section~\ref{sec:hardnessResults}, we complement these algorithmic results by showing that multilevel-planarity testing is \textsf{NP}-complete for abstract single-source graphs and for embedded multi-source graphs where it is~$\lvert\ell(v)\rvert \le 2$ for all~$v \in V$.
We finish with some concluding remarks in Section~\ref{sec:conclusion}.

\section{Preliminaries}
\label{sec:preliminaries}

This section consists of three parts.
First, we introduce basic terminology and notation.
Second, we discuss the relationship between multilevel planarity and existing planarity variants for directed graphs.
Third, we define a normal form for multilevel assignments, which simplifies the arguments in Sections~\ref{sec:embeddedSingleSourceGraphs} and~\ref{sec:cycles}.

\paragraph{Basic Terminology.}

Let~$G = (V, E)$ be a directed graph.
We use the terms \emph{drawing}, \emph{planar}, (\emph{combinatorial}) \emph{embedding} and \emph{face} as defined by Di Battista et al.~\cite{DiBattista1998Book}.
We say that two drawings are \emph{homeomorphic} if they respect the same combinatorial embedding.
A~\emph{multilevel assignment}~$\ell : V \to \mathcal P(\mathbb Z)$ assigns a finite set of possible integer levels to each vertex.
An upward-planar drawing is \emph{multilevel planar} if~$y(v) \in \ell(v)$ for all~$v \in V$.
Note that any finite set of integers can be represented as a finite list of finite integer intervals.
We choose this representation to be able to represent sets of integers that contain large intervals of numbers more efficiently.

A vertex of a directed graph with no incoming (outgoing) edges is a \emph{source} (\emph{sink}).
A directed, acyclic and planar graph with a single source~$s$ is an \emph{$sT$-graph}.
An~$sT$-graph with a single sink~$t$ and an edge~$(s,t)$ is an \emph{$st$-graph}.
In any upward-planar drawing of an~$st$-graph, the unique source and sink are the lowest and highest vertices, respectively, and both are incident to the outer face.
For a face~$f$ of a planar drawing, an incident vertex~$v$ is called \emph{source switch} (\emph{sink switch}) if all edges incident to~$f$ and~$v$ are outgoing (incoming).
Note that a source switch or sink switch does not need to be a source or sink in~$G$.
We will frequently add incoming edges to sources and outgoing edges to sinks during later constructions, referring to this as \emph{source canceling} and \emph{sink canceling}, respectively.
An \emph{oriented path} of length~$k$ is a sequence of vertices~$(v_1, v_2, \ldots, v_{k + 1})$ such that for all~$1 \leq i \leq k$ either the edge~$(v_i, v_{i + 1})$ or the edge~$(v_{i + 1}, v_i)$ exists.
A \emph{directed path} of length~$k$ is a sequence of vertices~$(v_1, v_2, \ldots, v_{k + 1})$ such that for all~$1 \leq i \leq k$ the edge~$(v_i, v_{i + 1})$ exists.
Let~$u, v \in V$ be two distinct vertices.
Vertex~$u$ is a \emph{descendant} of~$v$ in~$G$, if there exists a directed path from~$v$ to~$u$.
A \emph{topological ordering} is a function~$\tau : V \to \mathbb N$ such that for every~$v \in V$ and for each descendant~$u$ of~$v$ it is~$\tau(v) < \tau(u)$.

\paragraph{Relationship to Existing Planarity Variants.}

Multilevel-planarity testing is a generalization of level planarity.
To see this, let~$G = (V, E)$ be a directed graph together with a level assignment~$\gamma: V \to \mathbb{Z}$.
Define~$\ell(v) = \{\gamma(v)\}$ for all~$v \in V$.
It is readily observed that a drawing~$\Gamma$ of~$G$ is level planar with respect to~$\gamma$ if and only if~$\Gamma$ is multilevel planar with respect to~$\ell$.
Therefore, level planarity reduces to multilevel planarity in linear time.

Multilevel-planarity testing is also a generalization of upward planarity.
Without loss of generality, the vertices in an upward-planar drawing can be assigned integer~$y$-coordinates, and there is at least one vertex on each level.
Hence, upward planarity of~$G$ can be tested by setting~$\ell(v) = [1, \lvert V \rvert]$ for all~$v \in V$ and testing the multilevel planarity of~$G$ with respect to~$\ell$.
Therefore, upward planarity reduces to multilevel planarity in linear time.
By then restricting the multilevel assignment, multilevel planarity can also be seen as a constrained version of upward planarity.
Garg and Tamassia~\cite{Garg2002} showed the \textsf{NP}-completeness of upward-planarity testing, which directly gives the following.%
\begin{theorem}
  \label{thm:multiLevelPlanarityNPComplete}
  Multilevel-planarity testing is \textsf{NP}-complete.
\end{theorem}

\paragraph{Multilevel Assignment Normal Form.}

A multilevel assignment~$\ell$ has \emph{normal form} if for all~$(u, v) \in E$ it holds that~$\min\ell(u) < \min\ell(v)$ and~$\max\ell(u) < \max\ell(v)$.
Some proofs are easier to follow for multilevel assignments in normal form.
The following lemma justifies that we may assume without loss of generality that~$\ell$ has normal form.%
\begin{lemma}
  \label{lem:multiLevelAssignmentToNormalForm}
  Let~$G = (V, E)$ be a directed graph together with a multilevel assignment~$\ell$.
  Then there exists a multilevel assignment~$\ell'$ in normal form such that any drawing of~$G$ is multilevel planar with respect to~$\ell$ if and only if it is multilevel planar with respect to~$\ell'$.
  Moreover,~$\ell'$ can be computed in linear time.
\end{lemma}%
\begin{proof}
  The idea is to convert~$\ell(v)$ into~$\ell'(v) \subseteq \ell(v)$ for all~$v \in V$ by finding a lower bound~$l_v$ and an upper bound~$u_v$ for the level of~$v$, and then setting~$\ell'(v) = \ell(v) \cap [l_v, u_v]$.
  To find the lower bound, iterate over the vertices in increasing order with respect to some topological ordering~$\tau$ of~$G$.
  Because all edges have to be drawn as strictly~$y$-monotone curves, for each vertex~$v \in V$ it must be~$y(v) > \max_{(w, v) \in E} l_w$.
  So, set~$l_v = \max(\min\ell(v), \max_{(w, v) \in E} l_w + 1)$.
  Analogously, to find the upper bound, iterate over~$V$ in decreasing order with respect to~$\tau$.
  Again, because edges are drawn as strictly~$y$-monotone curves, for each vertex~$v \in V$ it must hold true that~$y(v) < \min_{(v, w) \in E} u_w$.
  Therefore, set~$u_v = \min(\max\ell(v), \min_{(v, w) \in E} u_w - 1)$.
  This means that in any multilevel-planar drawing of~$G$ the~$y$-coordinate of~$v \in V$ is~$y(v) \in \ell(v) \cap [l_v, u_v]$.
  So it follows that a drawing of~$G$ is multilevel planar with respect to~$\ell$ if and only if it is multilevel planar with respect to~$\ell'$.

  To see that the running time is linear, note that a topological ordering of~$G$ can be computed in linear time and every vertex and edge is handled at most twice during the procedure described above.
  Because every level candidate in~$\ell$ is removed at most once, the total running time is~$\mathcal O(n + \sum_{v \in V} \lvert\ell(v)\rvert)$, i.e., linear in the size of the input.
\end{proof}

\section{Embedded~$sT$-Graphs}
\label{sec:embeddedSingleSourceGraphs}

In this section, we characterize multilevel-planar $sT$-graphs as subgraphs of certain planar~$st$-graphs.
Similar characterizations exist for upward planarity and level planarity~\cite{DiBattista1988,Leipert1998}.
The idea behind our characterization is that for any given multilevel-planar drawing, we can find a set of edges that can be inserted without rendering the drawing invalid, and which make the underlying graph an~$st$-graph.
Thus, the graph must have been a subgraph of an~$st$-graph.
This technique is similar to the one found by Bertolazzi et al.~\cite{Bertolazzi1998}, and in fact is built on top of it.

To use this characterization for multilevel-planarity testing, we cannot require a multilevel-planar drawing to be given.
We show that if we choose the set of edges to be inserted carefully, the respective set of edges can be inserted into \emph{any} multilevel-planar drawing for a fixed combinatorial embedding.
An algorithm constructing such an edge set can therefore be used to test for multilevel planarity of embedded~$sT$-graphs, resulting in Theorem~\ref{thm:mlp-testing}.
The algorithm is constructive in the sense that it finds a multilevel-planar drawing, if it exists.
In Section~\ref{sec:hardnessResults}, we show that testing multilevel planarity of~$sT$-graphs without a fixed combinatorial embedding is \textsf{NP}-hard.
Recall that every multilevel-planar drawing is upward planar.
We now prove that the vertex with the largest~$y$-coordinate on the boundary of each face is the same across all homeomorphic drawings.%
\begin{lemma}
  \label{lem:uniqueSinkSwitch}
  Let~$G = (V, E)$ be a biconnected~$sT$-graph together with an upward-planar drawing~$\Gamma$.
  For each inner face~$f$ of~$\Gamma$ and each vertex~$v$ incident to~$f$, let~$\angle_{\Gamma, f}(v)$ denote the angle defined by the two edges incident to~$v$ and~$f$ in~$\Gamma$.
  Then the following properties hold:
  \begin{enumerate}
    \item There is exactly one sink switch~$t_f$ on the boundary of~$f$ with~$\angle_{\Gamma, f}(t_f) \le \pi$, namely the vertex with greatest~$y$-coordinate among all vertices incident to~$f$.
    \item Let~$\Gamma'$ be an upward-planar drawing of~$G$ that is homeomorphic to~$\Gamma$.
    Then the vertex~$t_f$ has the greatest~$y$-coordinate of all vertices incident to~$f$ in~$\Gamma'$.
  \end{enumerate}
\end{lemma}%
\begin{proof}
  The first property was observed by Bertolazzi et al.~\cite[page~138, fact~3]{Bertolazzi1998}.
  To prove the second property, assume that there exists an upward-planar drawing~$\Gamma'$ of~$G$ and a face~$f$ such that in~$\Gamma'$, vertex $t_f$ does not have the greatest~$y$-coordinate of all vertices incident to~$f$.
  Let~$e_1 = (v_1, t_f)$ and~$e_2 = (v_2, t_f)$ be the edges incident to~$f$ and~$t_f$.
  Figure~\ref{fig:uniqueSinkSwitch} shows exemplary drawings~$\Gamma$ and~$\Gamma'$.
  Because~$G$ has a single source~$s$, there exist directed paths~$p_1$ and~$p_2$ from~$s$ to~$v_1$ and~$v_2$, respectively.
  Then the left-to-right order of the edges~$e_1$ and~$e_2$ in~$\Gamma$ and~$\Gamma'$ is determined by the order of the outgoing edges at the last common vertex~$c$ on~$p_1$ and~$p_2$.
  Let~$t' \neq t_f$ be the vertex with greatest~$y$-coordinate of all vertices incident to~$f$ in $\Gamma'$.
  Then it holds that~$\angle_{\Gamma', f}(t') \le \pi$ and from the first property it follows that~$\angle_{\Gamma', f}(t_f) > \pi$.
  Since~$\Gamma$ and~$\Gamma'$ have the same underlying combinatorial embedding, the clockwise cyclic walk around~$f$ is identical in both drawings.
  But because~$\angle_{\Gamma, f}(t_f) \le \pi$ and~$\angle_{\Gamma', f}(t_f) > \pi$, the order of the outgoing edges of~$c$ is different in~$\Gamma$ and~$\Gamma'$.
  Note that~$c$ either has an incoming edge or it is~$s = c$, in which case the edge~$(s, t)$ lies to the left, i.e., the cyclic order of the edges around~$c$ is different in~$\Gamma$ and~$\Gamma'$.
  Therefore,~$\Gamma$ and~$\Gamma'$ are not homeomorphic.
\end{proof}

Note  that the result of Lemma~\ref{lem:uniqueSinkSwitch} also holds for embedded $sT$-graphs that are not biconnected.
Obviously it holds for any biconnected component.
Any subgraph~$G'$ that does not belong to any biconnected component is an attached tree inside a face $f$ given by the combinatorial embedding.
If~$f$ is an inner face, the unique vertex~$t_f$ of that face with maximal $y$-coordinate must be higher than any vertex of~$G'$ in any upward planar drawing. 

\begin{figure}[t]
  \begin{minipage}[t]{.52\linewidth}
    \centering
    \includegraphics{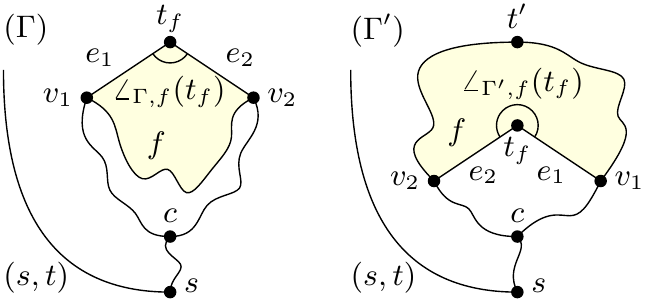}
    \caption{
      Proof of Lemma~\ref{lem:uniqueSinkSwitch}.
    }
    \label{fig:uniqueSinkSwitch}
  \end{minipage}%
  \hfill
  \begin{minipage}[t]{.4\linewidth}
    \centering
    \includegraphics{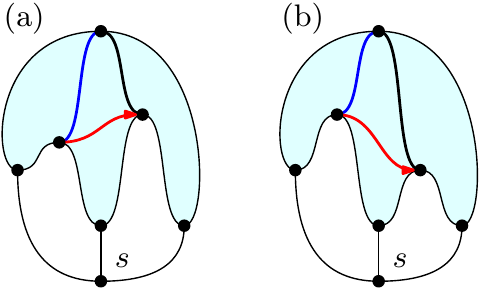}
    \caption{
      Not all edges are valid for the augmentation in Lemma~\ref{lem:stAugmentation}.
    }
    \label{fig:stAugmentation}
  \end{minipage}
\end{figure}

Bertolazzi et al.\ showed that any~$sT$-graph with an upward-planar embedding can be extended to an~$st$-graph with an upward-planar embedding that extends the original embedding~\cite{Bertolazzi1994,Bertolazzi1998}.
More formally, let~$G = (V, E)$ be an~$sT$-graph together with an upward-planar drawing~$\Gamma$.
Then there exists an~$st$-graph~$G_{st} = (V\,\dot\cup\,\{t\}, E\,\dot\cup\,E_{st})$ where~$t$ is the unique sink together with an upward-planar drawing~$\Gamma_{st}$ that extends~$\Gamma$.
Moreover,~$G_{st}$ and~$\Gamma_{st}$ can be computed in linear time.
Note that in general it is possible for a given~$E_{st}$ to choose an upward-planar drawing~$\Gamma$ of~$G$ so that the additional edges in~$E_{st}$ cannot be added into~$\Gamma$ as~$y$-monotone curves.
For an example, see Figure~\ref{fig:stAugmentation}, where augmenting with the red and black edge works only for the drawing shown in~(a), whereas augmenting with the blue and black edge works for both drawings.
In Lemma~\ref{lem:stAugmentation} we therefore show that there is a set~$E_{st}$ that can be added into any drawing with the same combinatorial embedding as~$\Gamma$.
In a way, this is the most general set~$E_{st}$.%
\begin{lemma}
  \label{lem:stAugmentation}
  Let~$G = (V, E)$ be a directed~$sT$-graph with a fixed combinatorial embedding.
  Then there exists an~$st$-graph~$G_{st} = (V\,\dot\cup\,\{t\}, E\,\dot\cup\,E_{st})$, where~$t$ is the unique sink, such that for any upward-planar drawing~$\Gamma$ of~$G$ there exists an upward-planar drawing~$\Gamma_{st}$ of~$G_{st}$ that extends~$\Gamma$.
  Moreover,~$G_{st}$ can be computed in linear time.
\end{lemma}%
\begin{proof}
  Start by finding an initial upward-planar drawing~$\Gamma_{init}$ of~$G$ in linear time using the algorithm due to Bertolazzi et al.~\cite{Bertolazzi1998}.
  The algorithm also outputs a matching~$st$-graph~$G_{st}$ together with an upward-planar drawing~$\Gamma_{st}$ that extends~$\Gamma_{init}$.
  Note that any edge~$e \in E_{st}$ is drawn within some face of~$\Gamma_{init}$.
  Because~$t$ is the only sink of~$G_{st}$, it must have the highest~$y$-coordinate among all vertices in every upward-planar drawing of~$G_{st}$.
  Therefore, changing all edges in~$E_{st}$ drawn within the outer face to have endpoint~$t$ ensures that these edges can be drawn within the outer face of any upward-planar drawing~$\Gamma$ of~$G$ as~$y$-monotone curves while preserving planarity.
  For any inner face~$f$, Lemma~\ref{lem:uniqueSinkSwitch} states that there is a unique~$t_f$ incident to~$f$ with greatest~$y$-coordinate in every upward-planar drawing of~$G$ homeomorphic to~$\Gamma$.
  So changing all edges in~$E_{st}$ that are drawn within~$f$ to have endpoint~$t_f$ ensures that these edges can be drawn within~$f$ in any upward-planar drawing~$\Gamma$ of~$G$ as~$y$-monotone curves while preserving planarity.
  By precomputing~$t_f$ for every face, this procedure handles every edge in~$E_{st}$ in constant time, which gives linear running time overall.
\end{proof}

We now have a set of edges that can be used to complete~$G$ into~$G_{st}$. If a multilevel-planar drawing for the given combinatorial embedding of~$G$ respecting~$\ell$ exists, then it must also exist for~$G_{st}$.
However, the property of~$\ell$ being in normal form might not be fulfilled anymore in~$G_{st}$ because of the added edges.
We therefore need to bring~$\ell$ into normal form~$\ell'$ again.
Lemma~\ref{lem:multiLevelAssignmentToNormalForm} tells us that this does not impact multilevel planarity.
We conclude that~$G$ is multilevel planar with respect to~$\ell$ if and only if~$G_{st}$ is multilevel planar with respect to~$\ell'$.
The final property we need is proved by Leipert~\cite[page~117, Theorem~5.1]{Leipert1998}, and described in an article by J\"unger and Leipert~\cite{Juenger1999}.%
\begin{lemma}
  \label{lem:stGraphLevelPlanar}
  Let~$G$ be an~$st$-graph together with a level assignment~$\gamma$.
  Then for any combinatorial embedding of~$G$ there exists a drawing of~$G$ with that embedding that is level planar with respect to~$\gamma$.
\end{lemma}

\noindent If~$\ell'$ is in normal form, ~$\ell'(v) \neq \emptyset$ is a necessary and sufficient condition that there exists a level assignment~$\gamma: V \to \mathbb Z$ with~$\gamma(v) \in \ell'(v)$ for all~$v \in V$.
Setting~$\gamma(v) = \min \ell'(v)$ is one possible such level assignment.
Then~$G$ is level planar with respect to~$\gamma$ and therefore multilevel planar with respect to~$\ell$, resulting in the characterization of multilevel-planar $st$-graphs:%
\begin{corollary}
	\label{cor:stGraphMultiLevelPlanar}
	Let~$G$ be an~$st$-graph together with a multilevel assignment~$\ell$ in normal form.
	Then there exists a multilevel-planar drawing for any combinatorial embedding of~$G$ if and only if~$\ell(v) \neq \emptyset$ for all~$v$.
\end{corollary}

\noindent For a constructive multilevel-planarity testing algorithm, we now first take the edge set computed by the algorithm by Bertolazzi et al.~\cite{Bertolazzi1998} and modify it using Lemma~\ref{lem:stAugmentation} to complete any~$sT$-graph to an~$st$-graph.
Note that for this step, we need a fixed combinatorial embedding to be given, as is required by Point 2 of Lemma~\ref{lem:uniqueSinkSwitch}.
Once arrived at an~$st$-graph, we only need to check the premise of Corollary~\ref{cor:stGraphMultiLevelPlanar}.
This concludes the testing algorithm:%
\begin{theorem}
  \label{thm:mlp-testing}
  Let~$G$ be an embedded~$sT$-graph with a multilevel assignment~$\ell$.
  Then it can be decided in linear time whether there exists a multilevel-planar drawing of~$G$ respecting that embedding.
  If so, such a drawing can be computed within the same running time.
\end{theorem}

\noindent Our algorithm uses the fact that to augment~$sT$-graphs to~$st$-graphs, only edges connecting sinks to other vertices need to be inserted.
For graphs with multiple sources and multiple sinks, further edges connecting sources to other vertices need to be inserted.
The interactions that occur then are very complex:
In Section~\ref{sec:hardnessResults}, we show that deciding multilevel planarity is \textsf{NP}-complete for embedded multi-source graphs.
In the next section, we identify oriented cycles as a class of multi-source graphs for which multilevel planarity can be efficiently decided.

\section{Oriented Cycles}
\label{sec:cycles}

In this section, we present a constructive multilevel-planarity testing algorithm for oriented cycles, i.e., directed graphs whose underlying undirected graph is a simple cycle.
We start by giving a condition for when an oriented cycle~$G = (V, E)$ together with some level assignment~$\gamma$ admits a level-planar drawing.
This condition yields an algorithm for the multilevel-planar setting.

In this section,~$\gamma$ is always level assignment and~$\ell$ is always a multilevel assignment.
Define~$\max\gamma = \max\{\gamma(v) \mid v \in V\}$ and $\min\gamma = \min\{\gamma(v) \mid v \in V\}$.
Further set~$\max\ell = \max\{\max\ell(v) \mid v \in V\}$ and~$\min\ell = \min\{\min\ell(v) \mid v \in V\}$.
Let~$S_{\min} \subset V$ be sources with minimal level, i.e.,~$S_{\min} = \{v \in V \mid \gamma(v) = \min\gamma\}$, and let~$T_{\max} \subset V$ be the sinks with maximal level.
We call sources in~$S_{\min}$ \emph{minimal sources}, sinks in~$T_{\max}$ are \emph{maximal sinks}.
Two sinks~$t_1, t_2 \in T_{\max}$ are \emph{consecutive} if there is an oriented path between~$t_1$ and~$t_2$ that does not contain any vertex in~$S_{\min}$.
The set~$T_{\max}$ is \emph{consecutive} if all sinks in~$T_{\max}$ are pairwise consecutive.
We define consecutiveness for sources in~$S_{\min}$ analogously.
Because~$G$ is a cycle, consecutiveness of~$T_{\max}$ also means that~$S_{\min}$ is consecutive.
If both~$S_{\min}$ and~$T_{\max}$ are consecutive, we say that~$\gamma$ is \emph{separating}.%
\begin{lemma}
  \label{lem:cycleEquivalence}
  Let~$G$ be an oriented cycle with a level assignment~$\gamma$.
  Then~$G$ is level planar with respect to~$\gamma$ if and only if~$T_{\max}$ is consecutive.
\end{lemma}

\begin{proof}
  For the ``if'' part, augment~$G$ to a planar~$st$-graph as follows.
  Let~$p_t$ be the oriented path of minimal length that contains all sinks in~$T_{\max}$ and no vertex in~$S_{\min}$, and let~$t_1, t_2 \in T_{\max}$ denote its endpoints.
  Let~$p_s$ be the oriented path from~$t_2$ to~$t_1$ so that~$p_s \cup p_t = G$ and~$p_s \cap p_t = \{t_1, t_2\}$.
  Draw~$p_t$ from left to right; see Fig.~\ref{fig:cycleLevelPlanar}.
  Below it, draw~$p_s$ from right to left.
  Fix some vertex~$s_{\min} \in S_{\min}$ and add an edge from~$s_{\min}$ to every source on the path~$p_t$.
  Add a new vertex~$s$ to~$G$, set~$\gamma(s) = \min\gamma - 1$ and add an edge from~$s$ to every source on the path~$p_s$.
  Thus,~$s$ is now the only source.
  Next, observe that any sink~$t_s$ on the path~$p_s$ is drawn to the left of~$s_{\min}$ or to the right of~$s_{\min}$.
  Add the edge~$(t_s, t_1)$ or the edge~$(t_s, t_2)$, respectively.
  Finally, add a new vertex~$t$ to~$G$, set~$\gamma(t) = \max\gamma + 1$ and add an edge from every sink on the path~$p_t$ to~$t$.
  Thus,~$t$ is now the only sink.
  All added edges~$(u,v)$ satisfy~$\gamma(u) < \gamma(v)$.
  Hence,~$G$ is now an~$st$-graph with a level assignment~$\gamma$ and so Lemma~\ref{lem:stGraphLevelPlanar} gives that~$G$ is level planar with respect to~$\gamma$.

  For the ``only if'' part, assume that~$T_{\max}$ is not consecutive.
  Then there are maximal sinks~$t_1, t_2 \in T_{\max}$ and minimal sources~$s_1, s_2 \in S_{\min}$ that appear in the order~$s_1, t_1, s_2, t_2$ around the cycle underlying~$G$.
  Because the chosen sinks and sources are highest and lowest vertices, respectively, the four edge-disjoint paths that connect them must intersect.
\end{proof}

\begin{figure}[t]
  \centering
  \includegraphics{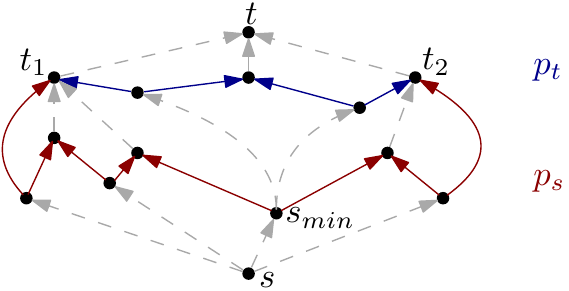}
  \caption{
    An~$st$-augmentation of an oriented cycle.
    The gray dashed edges are added for the~$st$-augmentation.
    The blue edges belong to path~$p_t$ and the red edges to path~$p_s$.
  }
  \label{fig:cycleLevelPlanar}
\end{figure}

Recall that any multilevel-planar drawing is a level-planar drawing with respect to some level assignment~$\gamma$.
Lemma~\ref{lem:cycleEquivalence} gives a necessary and sufficient condition for~$\gamma$ so that the drawing is level planar.
Given a multilevel assignment~$\ell$, we therefore find an induced separating level assignment~$\gamma$, or determine that no such level assignment exists.
It must be~$\ell(v) \neq \emptyset$ for all~$v \in V$; otherwise,~$G$ admits no multilevel drawing.
We find an induced level assignment~$\gamma$ that keeps the sets~$S_{\min}$ and~$T_{\max}$ as small as possible, because such a level assignment is, intuitively, most likely to be separating.
To this end, let~$S' \subset V$ denote the sources of~$G$ such that for~$s' \in S'$ we have~$\min\ell(s') = \min\ell$.
Further, let~$S'' \subseteq S'$ denote the sources of~$G$ such that for~$s'' \in S''$ we have~$\ell(s'') = \{\min\ell\}$.
Likewise, let~$T' \subset V$ denote the sinks of~$G$ such that for each~$t' \in T'$ it holds that~$\max\ell(t') = \max\ell$ let~$T'' \subseteq T'$ denote the sinks of~$G$ such that for~$t'' \in T''$ we have~$\ell(t'') = \{\max\ell\}$.

Suppose~$S'' \neq \emptyset$.
Observe that due to the multilevel assignment, all sources in~$S''$ have to be minimal sources.
Therefore, set~$S_{\min} = S''$.
Otherwise, if~$S'' = \emptyset$, pick any source~$s' \in S'$ and set~$S_{\min} = \{s'\}$.
Proceed analogously to find~$T_{\max}$.
If~$T'' \neq \emptyset$, set~$T_{\max} = T''$.
Otherwise, pick any sink~$t' \in T'$ and set~$T_{\max} = \{t'\}$.
Note that if~$S''$ or~$T''$ are not empty there is no choice but to add all sources or sinks in them to~$S_{\min}$ or~$T_{\max}$.
Otherwise~$S_{\min}$ or~$T_{\max}$ contains only one vertex, which guarantees that~$T_{\max}$ is consecutive.
Since~$\ell$ is in normal form, any remaining vertex can be assigned greedily to its minimum possible level above all its ancestors.
Hence,~$G$ is multilevel planar with respect to~$\ell$ if and only if~$T_{\max}$ is consecutive.
We conclude the following.%
\begin{theorem}
	\label{thm:cycle-testing}
  Let~$G$ be an oriented cycle together with a multilevel assignment~$\ell$.
  Then it can be decided in linear time whether~$G$ admits a drawing that is multilevel planar with respect to~$\ell$.
  Furthermore, if such a drawing exists, it can be computed within the same time.
\end{theorem}

\section{Hardness Results}
\label{sec:hardnessResults}

We now show that multilevel-planarity testing is~\textsf{NP}-complete even in very restricted cases, namely for~$sT$-graphs without a fixed embedding and for embedded multi-source graphs with at most two possible levels for each vertex.

\subsection{$sT$-Graphs with Variable Embedding}

In Section~\ref{sec:embeddedSingleSourceGraphs}, we showed that testing multilevel planarity of embedded~$sT$-graphs is feasible in linear time, because for every inner sink there is a unique sink switch to cancel it with.
We now show that dropping the requirement that the embedding is fixed makes multilevel-planarity testing \textsf{NP}-hard.
To this end, we reduce the \textsc{scheduling with release times and deadlines (Srtd)} problem, which is strongly \textsf{NP}-complete~\cite{GareyJohnson1977}, to multilevel-planarity testing.
An instance of this scheduling problem consists of a set of tasks~$T = \{t_1, \ldots, t_n\}$ with individual release times~$r_1, \ldots, r_n \in \mathbb N$, deadlines~$d_1, \ldots, d_n \in \mathbb N$ and processing times~$p_1, \ldots, p_n \in \mathbb N$ for each task (we assume~$0 \not\in \mathbb{N}$), where~$\sum_{i=1}^n p_i$ is bounded by a polynomial in~$n$.
See Fig.~\ref{fig:singleSourceNPComplete} (a) for an example.
The question is whether there is a non-preemptive schedule~$\sigma : T \to \mathbb N$, such that for each~$i \in \{1, \ldots, n\}$ we get
\begin{enumerate*}[label=(\arabic*)]
  \item~$\sigma(t_i) \geq r_i$, i.e., no task starts before its release time,
  \item~$\sigma(t_i) + p_i \leq d_i$, i.e., each task fi\-ni\-shes before its deadline, and
  \item~$\sigma(t_i) < \sigma(t_j) \implies \sigma(t_i) + p_i \leq \sigma(t_j)$ for any~$j \in \{1, \ldots, n\} \setminus \{i\}$, i.e., no two tasks are executed at the same time.
\end{enumerate*}

\begin{figure}[t]
  \centering
  \includegraphics{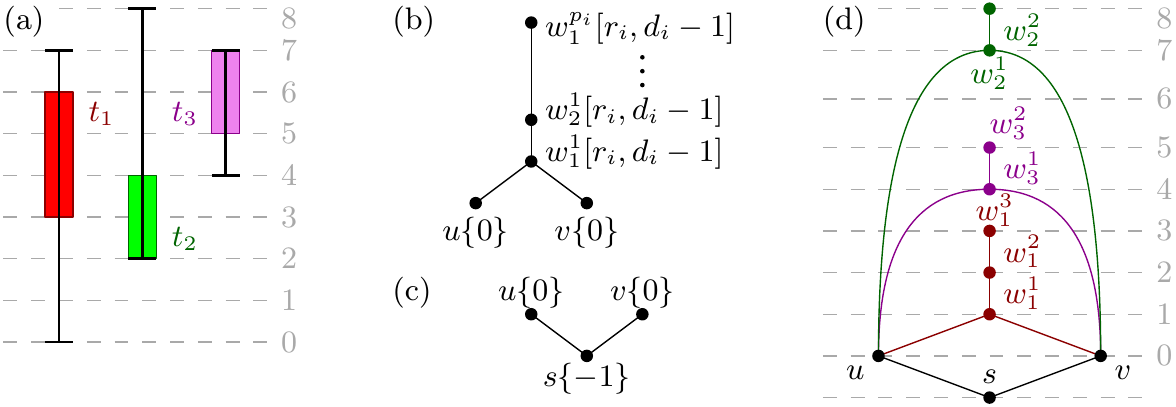}
  \caption{
    A task gadget~(b) for each task and one base gadget~(c) that provides the single source are used to turn a~\textsc{Srtd} instance~(a) into a multilevel-planarity testing instance~(d).
  }
  \label{fig:singleSourceNPComplete}
\end{figure}
Create for every task~$t_i \in T$ a \emph{task gadget}~$\mathcal T_i$ that consists of two vertices~$u$,~$v$ together with a directed path~$P_i = (w_i^1, w_i^2, \ldots, w_i^{p_i})$ of length~$p_i - 1$; see Fig.~\ref{fig:singleSourceNPComplete}~(b).
For each vertex~$w_i^j$ on~$P_i$ set~$\ell(w_i^j) = [r_i,d_i - 1]$, i.e., all possible points of time at which this task can be executed.
Set~$\ell(u) = \ell(v) = \{0\}$.
Join all task gadgets with a \emph{base gadget}.
The base gadget consists of three vertices~$s, u, v$ and two edges~$(s, u), (s, v)$, where~$u$ is placed to the left of~$v$; see Fig.~\ref{fig:singleSourceNPComplete}~(c).
Set~$\ell(s) = \{-1\}$ and, again, set~$\ell(u) = \ell(v) = \{0\}$.
Merge all gadgets at their common vertices~$u$ and~$v$; see Fig.~\ref{fig:singleSourceNPComplete}~(d).
Because \textsc{Srtd} is strongly \textsf{NP}-complete, the size of the resulting graph is polynomial in the size of the input.
Further, because the task gadgets may not intersect in a planar drawing and because they are merged at their common vertices~$u$ and~$v$, they are stacked on top of each other, inducing a valid schedule of the associated tasks. Contrasting linear-time tests of upward planarity and level planarity for~$sT$-graphs we conclude:%
\begin{theorem}
  \label{thm:singleSourceNPComplete}
  Let~$G$ be an~$sT$-graph together with a multilevel assignment~$\ell$.
  Testing whether~$G$ is multilevel planar with respect to~$\ell$ is \textsf{NP}-complete.
\end{theorem}%
Using a very similar reduction one can also show \textsf{NP}-completeness of multilevel-planarity testing for trees.
\todo{Link to full proof and arxiv version here.}

\subsection{Embedded Multi-Source Graphs}

\begin{figure}[t]
  \centering
  \includegraphics{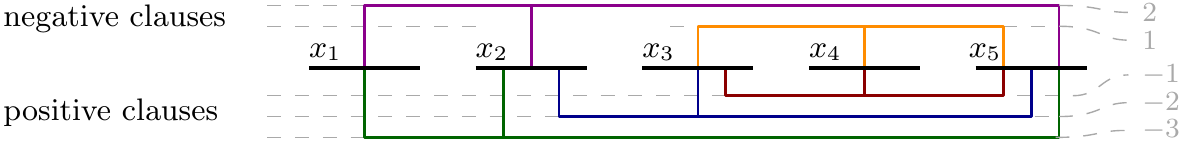}
  \caption{
      A rectilinear embedding of the planar monotone \textsc{3-Sat} instance $\textcolor{darkgreen}{(x_1 \lor x_2 \lor x_5)} \land \textcolor{darkblue}{(x_2 \lor x_3 \lor x_5)} \land \textcolor{darkred}{(x_3 \lor x_4 \lor x_5)} \land \textcolor{darkmagenta}{(\neg x_1 \lor \neg x_2 \lor \neg x_5)} \land \textcolor{darkorange}{(\neg x_3 \lor \neg x_4 \lor \neg x_5)}$.
  }
  \label{fig:planarMonotone3SATInstance}
\end{figure}
We show that multilevel-planarity testing for embedded directed graphs is \textsf{NP}-complete by reducing from \textsc{planar monotone 3-Sat}~\cite{DeBerg2012}.
An instance~$\mathcal I = (\mathcal V, \mathcal C)$ of this problem is a \textsc{3-Sat} instance with variables~$\mathcal V$, clauses~$\mathcal C$ and additional restrictions.
Namely, each clause is \emph{monotone}, i.e., it is either positive or negative, meaning that it consists of either only positive or only negative literals, respectively.
The \emph{variable-clause graph} of~$\mathcal I$ consists of the nodes~$\mathcal V \cup \mathcal C$ connected by an arc if one of the nodes is a variable and the other node is a clause that uses this variable.
The variable-clause graph can be drawn such that all variables lie on a horizontal straight line, positive and negative clauses are drawn as horizontal line segments with integer~$y$-coordinates below and above that line, respectively, and arcs connecting clauses and variables are drawn as non-intersecting vertical line segments; see Fig.~\ref{fig:planarMonotone3SATInstance}.
We call this a \emph{planar rectilinear embedding} of~$\mathcal I$.

Let~$\Gamma_{\mathcal I}$ be a planar rectilinear embedding of~$\mathcal I$.
Transform this into a multilevel-planarity testing instance by replacing each positive or negative clause of~$\mathcal I$ with a positive or negative clause gadget and merging them at common vertices.
Fig.~\ref{fig:embeddingNPComplete}~(a) shows the gadget for the positive clause~$(x_a \lor x_b \lor x_c)$.
The vertices~$x_a$,~$x_b$ and~$x_c$ are variables in~$\mathcal V$.
We call vertex~$p_i$ the \emph{pendulum}.
A variable~$x \in \mathcal V$ is set to true (false) if it lies on level 1 (level 0).
In a positive clause gadget~$p_i$ must lie on level 0, and so it forces one variable to lie on level 1, i.e., be set to true.
The gadget for a negative clause~$(\neg x_a \lor \neg x_b \lor \neg x_c)$ works symmetrically; its pendulum forces one variable to lie on level 0, i.e., be set to false; see Fig.~\ref{fig:embeddingNPComplete}~(b).
\begin{figure}[t]
  \centering
  \includegraphics{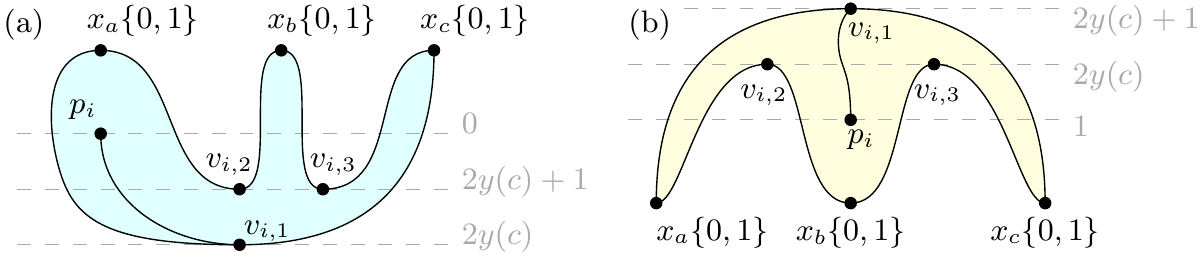}
  \caption{
      Gadgets for the clauses $(x_a \lor x_b \lor x_c)$~(a) and $(\neg x_a \lor \neg x_b \lor \neg x_c)$~(b).
  }
  \label{fig:embeddingNPComplete}
\end{figure}

\begin{theorem}
  \label{thm:embeddingNPComplete}
  Let~$G = (V, E)$ be an embedded directed graph together with a multilevel assignment~$\ell$.
  Testing whether~$G$ is multilevel planar is \textsf{NP}-complete, even if it is~$\rvert\ell(v)\rvert \le 2$ for all~$v \in V$.
\end{theorem}

\section{Conclusion}
\label{sec:conclusion}

\begin{table}[b]
  \begin{center}
  	\begin{tabular}{rcccccc}
  		\toprule

  		&
      \multicolumn{3}{c}{fixed combinatorial embedding} &
      \multicolumn{3}{c}{not embedded} \\

  		\cmidrule(r{3pt}){2-4}
      \cmidrule(l{3pt}){5-7}

  		&
      $st$-Graphs &
      $sT$-Graphs &
      arbitrary &
      Cycles &
      $sT$-Graphs &
      Trees \\

  		\midrule
  		\vspace{5pt}

  		Upward Planarity &
      $O(1)$~\cite{Bertolazzi1994} &
      $O(n)$~\cite{Bertolazzi1994} &
      \textsf{P}~\cite{Bertolazzi1994} &
      $O(n)$~\cite{Bertolazzi1998} &
      $O(n)$~\cite{Bertolazzi1998} &
      $O(1)$~\cite{DiBattista1998Book} \\

  		\vspace{5pt}

  		Multilevel Planarity &
      \makecell{$O(1)$ \\ (Cor.~\ref{cor:stGraphMultiLevelPlanar})} &
  		\makecell{$O(n)$ \\ (Thm.~\ref{thm:mlp-testing})} &
  		\makecell{\textsf{NPC}\\ (Thm.~\ref{thm:embeddingNPComplete})} & 
  		\makecell{$O(n)$ \\ (Thm.~\ref{thm:cycle-testing})} & 
  		\makecell{\textsf{NPC}\\ (Thm.~\ref{thm:singleSourceNPComplete})} &
  		\makecell{\textsf{NPC}\\ (Thm.~\ref{thm:trees-hard})} \\	
  		
  		Level Planarity &
      $O(1)$~\cite{Juenger1999} &
      $O(n)$~\cite{Juenger1999} &
      ? &
      $O(n)$~\cite{Juenger1999} &
      $O(n)$~\cite{Juenger1999} &
      $O(n)$~\cite{Juenger1999} \\

  		\bottomrule    
  	\end{tabular}
  \end{center}
\end{table}

In this paper we introduced and studied the multilevel-planarity testing problem.
It is a generalization of both upward-planarity testing and level-planarity testing.

We started by giving a linear-time algorithm to decide multilevel planarity of embedded~$sT$-graphs.
The proof of correctness of this algorithm uses insights from both upward planarity and level planarity.
In opposition to this result, we showed that deciding the multilevel planarity of~$sT$-graphs without a fixed embedding is \textsf{NP}-complete.
This also contrasts the situation for upward planarity and level planarity, both of which can be decided in linear time for such graphs.

We also gave a linear-time algorithm to decide multilevel planarity of oriented cycles, which is interesting because the existence of multiple sources makes many related problems \textsf{NP}-complete, e.g., testing upward planarity, partial level planarity or ordered level planarity.
This positive result is contrasted by the fact that multilevel-planarity testing is \textsf{NP}-complete for oriented trees.
Whether multilevel-planarity testing becomes tractable for trees with a given combinatorial embedding remains an open question.
We also showed that deciding multilevel planarity remains \textsf{NP}-complete for embedded multi-source graphs where each vertex is assigned either to exactly one level, or to one of two adjacent levels.
This result again contrasts the existence of efficient algorithms for testing upward planarity and level planarity of embedded multi-source graphs.
The following table summarizes our results.

\bibliographystyle{splncs04}
\bibliography{strings_short,references}

\newcommand{\bibdac}[2]{DAC'#2} \newcommand{\bibinvisau}[1]{Proc. Australian
  Symp. Inf. Vis. (invis.au #1)} \newcommand{\bibieeepdp}[2]{Proc. #1 IEEE
  Symp. Par. Distr. Process. #2} \newcommand{\bibieeecs}[1]{Proc. IEEE Int.
  Symp. Circ. Syst. #1} \newcommand{\bibcccg}[2]{CCCG'#2}
  \newcommand{\bibswat}[2]{SWAT'#2} \newcommand{\bibipco}[2]{IPCO'#2}
  \newcommand{\bibsofsem}[2]{SOFSEM'#2} \newcommand{\bibstoc}[2]{STOC'#2}
  \newcommand{\bibfocs}[2]{FOCS'#2} \newcommand{\bibsoda}[2]{SODA'#2}
  \newcommand{\bibgd}[2]{GD'#2} \newcommand{\bibinfovis}[1]{InfoVis'#1}
  \newcommand{\bibvis}[1]{Vis'#1} \newcommand{\bibpvis}[1]{PacificVis'#1}
  \newcommand{\bibsoftvis}[2]{SoftVis'#2} \newcommand{\bibeurocg}[2]{EuroCG'#2}
  \newcommand{\bibsocg}[2]{SoCG'#2} \newcommand{\bibwads}[2]{WADS'#2}
  \newcommand{\bibwg}[2]{WG'#2} \newcommand{\bibgta}{Proceedings of the
  Conference at Graph Theory and Applications}
  \newcommand{\bibisaac}[2]{ISAAC'#2} \newcommand{\bibcocoon}[2]{COCOON'#2}
  \newcommand{\bibtamc}[2]{TAMC'#2} \newcommand{\bibicalp}[2]{ICALP'#2}
  \newcommand{\biblatin}[2]{LATIN'#2} \newcommand{\bibesa}[2]{ESA'#2}
\begin{thebibliography}{10}
\providecommand{\url}[1]{\texttt{#1}}
\providecommand{\urlprefix}{URL }
\providecommand{\doi}[1]{https://doi.org/#1}

\bibitem{Angelini2016}
Angelini, P., Da~Lozzo, G., Di~Battista, G., Frati, F., Patrignani, M., Rutter,
  I.: Beyond level planarity. In: Hu, Y., N{\"o}llenburg, M. (eds.)
  \bibgd{24th}{16}. pp. 482--495. Springer (2016)

\bibitem{Angelini2015-proper}
Angelini, P., Da~Lozzo, G., Di~Battista, G., Frati, F., Roselli, V.: The
  importance of being proper (in clustered-level planarity and $t$-level
  planarity). Theoretical Comput. Sci.  \textbf{571}, ~1--9 (2015)

\bibitem{Angelini2015}
Angelini, P., Di~Battista, G., Frati, F., Jel\'{\i}nek, V., Kratochv\'{\i}l,
  J., Patrignani, M., Rutter, I.: Testing planarity of partially embedded
  graphs. ACM Trans. Alg.  \textbf{11}(4),  32:1--32:42 (2015)

\bibitem{Bachmaier2005}
Bachmaier, C., Brandenburg, F.J., Forster, M.: Radial level planarity testing
  and embedding in linear time. J. Graph Alg. Appl.  \textbf{9}(1),  53--97
  (2005)

\bibitem{Bertolazzi1994}
Bertolazzi, P., Di~Battista, G., Liotta, G., Mannino, C.: Upward drawings of
  triconnected digraphs. Algorithmica  \textbf{12}(6),  476--497 (1994)

\bibitem{Bertolazzi1998}
Bertolazzi, P., Di~Battista, G., Mannino, C., Tamassia, R.: Optimal upward
  planarity testing of single-source digraphs. SIAM J. Comput.  \textbf{27}(1),
   132--169 (1998)

\bibitem{Brueckner2017}
Br\"uckner, G., Rutter, I.: Partial and constrained level planarity. In: Klein,
  P.N. (ed.) \bibsoda{28th}{17}. pp. 2000--2011 (2017)

\bibitem{DeBerg2012}
De~Berg, M., Khosravi, A.: Optimal binary space partitions for segments in the
  plane. Int. J. Comput. Geom. \& Appl.  \textbf{22}(3),  187--205 (2012)

\bibitem{DiBattista1998Book}
Di~Battista, G., Eades, P., Tamassia, R., Tollis, I.G.: Graph Drawing:
  Algorithms for the Visualization of Graphs. Prentice Hall PTR, 1st edn.
  (1998)

\bibitem{DiBattista2008}
Di~Battista, G., Frati, F.: Efficient c-planarity testing for embedded flat
  clustered graphs with small faces. In: Hong, S.H., Nishizeki, T., Quan, W.
  (eds.) \bibgd{16th}{08}. pp. 291--302. Springer (2008)

\bibitem{DiBattista1988}
Di~Battista, G., Tamassia, R.: Algorithms for plane representations of acyclic
  digraphs. Theoretical Comput. Sci.  \textbf{61}(2),  175--198 (1988)

\bibitem{Forster2004}
Forster, M., Bachmaier, C.: Clustered level planarity. In: Van Emde~Boas, P.,
  Pokorn{\'y}, J., Bielikov{\'a}, M., {\v{S}}tuller, J. (eds.)
  \bibsofsem{30th}{04}. pp. 218--228. Springer (2004)

\bibitem{GareyJohnson1977}
Garey, M.R., Johnson, D.S.: Two-processor scheduling with start-times and
  deadlines. SIAM J. Comput.  \textbf{6}(3),  416--426 (1977)

\bibitem{Garg2002}
Garg, A., Tamassia, R.: On the computational complexity of upward and
  rectilinear planarity testing. SIAM J. Comput.  \textbf{31}(2),  601--625
  (2002)

\bibitem{Harrigan2008}
Harrigan, M., Healy, P.: Practical level planarity testing and layout with
  embedding constraints. In: Hong, S.H., Nishizeki, T., Quan, W. (eds.)
  \bibgd{16th}{08}. pp. 62--68. Springer (2008)

\bibitem{Jelinek2013}
Jel\'inek, V., Kratochv\'il, J., Rutter, I.: A {Kuratowski}-type theorem for
  planarity of partially embedded graphs. Comput. Geom.: Theory Appl.
  \textbf{46}(4),  466--492 (2013)

\bibitem{Jelinkova2009}
Jel\'inkov\'a, E., K\'ara, J., Kratochv\'il, J., Pergel, M., Such\'y, O.,
  Vysko\v{c}il, T.: Clustered planarity: Small clusters in cycles and
  {Eulerian} graphs. J. Graph Alg. Appl.  \textbf{13}(3),  379--422 (2009)

\bibitem{Juenger1999}
J{\"u}nger, M., Leipert, S.: Level planar embedding in linear time. In:
  Kratochv\'{\i}l, J. (ed.) \bibgd{7th}{99}. pp. 72--81. Springer (1999)

\bibitem{Klemz2017}
Klemz, B., Rote, G.: Ordered level planarity, geodesic planarity and
  bi-monotonicity. In: Frati, F., Ma, K.L. (eds.) \bibgd{25th}{17}. pp.
  440--453. Springer (2018)

\bibitem{Leipert1998}
Leipert, S.: Level Planarity Testing and Embedding in Linear Time. Ph.D.
  thesis, University of Cologne (1998)

\end{thebibliography}

\newpage
\appendix

\section{Omitted Parts from Section~\ref{sec:hardnessResults}}
\label{apx:hardness-results}

\subsection{Proof of Theorem~\ref{thm:singleSourceNPComplete}}

\begin{proof}
  The graph~$G$ is multilevel planar if and only if there is a valid one-processor schedule for the \textsc{Srtd} instance.
  To see this, start with a valid schedule~$\sigma$.
  Define a level assignment~$\gamma$ as follows.
  Set~$\gamma(s) = -1$,~$\gamma(u) = \gamma(v) = 0$.
  And for~$1 \leq i \leq n$ and~$1 \leq j \leq p_i$, set~$\gamma(w_i^j) = \sigma(t_i) + j$.
  Since~$\sigma$ is non-preemptive, it induces a total order on the tasks, without loss of generality~$\sigma(t_1) < \ldots < \sigma(t_n)$.
  Order the edges to the task gadgets~$\mathcal T_1, \mathcal T_2, \ldots, \mathcal T_n$ from right to left at~$u$, and from left to right at~$v$.
  Observe that any sink of~$G$ is the endpoint of a directed path of a task gadget.
  For~$1 \leq i < n$ cancel the sink~$w_i^{p_i}$ by connecting it to~$w_{i + 1}^1$.
  This is possible because the schedule is valid.
  Then Lemma~\ref{lem:stGraphLevelPlanar} gives that there exists a drawing of~$G$ that is level-planar with respect to~$\gamma$.
  Because it is~$\gamma(v) \in \ell(v)$ for all~$v \in V$ by construction,~$G$ is indeed multilevel planar with respect to~$\ell$.

  For the reverse direction, consider a drawing~$\Gamma$ of~$G$ that is multilevel planar with respect to~$\ell$.
  Let~$\gamma$ denote the level assignment induced by~$\Gamma$.
  Lemma~\ref{lem:stAugmentation} gives that~$G$ and~$\Gamma$ can be augmented to an~$st$-supergraph of~$G$ with a level-planar drawing~$\Gamma_{st}$ that extends~$\Gamma$.
  For~$1 \leq i < n$ this means that the sink~$w_i^{p_i}$ of~$G$, is canceled at some vertex~$x$.
  Because~$w_i^{p_i}$ has degree one, it is incident to only one face~$f$.
  Note that because~$w_i^{p_i}$ is the highest vertex of~$\mathcal T_i$, it cannot be canceled at a vertex that belongs to~$\mathcal T_i$.
  The only vertex incident to~$f$ that does not belong to~$\mathcal T_i$ is the vertex~$w_{i + 1}^1$.
  Therefore, it must be~$\gamma(w_i^{p_i}) < \gamma(w_{i + 1}^1)$.
  Now set~$\sigma(t_i) = \gamma(w_i^1)$ for~$1 \leq i \leq n$.
  By the argument we just made it is~$\sigma(t_i) + p_i < \sigma(t_{i + 1})$.
  Moreover,~$\sigma(t_i) \ge r_i$ and~$\sigma(t_i) + p_i \leq d_i$ is ensured by the multilevel assignment.
  Hence,~$\sigma$ is a valid schedule.
\end{proof}

\subsection{Oriented Trees}

We can show~\textsf{NP}-completeness of oriented trees with a very similar reduction as for $sT$-graphs without a fixed embedding.
The needed gadgets are only slightly different.

\begin{theorem}
	\label{thm:trees-hard}
  Let~$T$ be an oriented tree together with a multilevel assignment~$\ell$.
  Testing whether~$T$ is multilevel planar with respect to~$\ell$ is \textsf{NP}-complete.
\end{theorem}
\begin{proof}
  As in the proof for Theorem~\ref{thm:singleSourceNPComplete} we reduce from \textsc{Srtd}.
  Let~$T = \{t_1, \ldots, t_n\}$, $r_1, \ldots, r_n$, $d_1, \ldots, d_n$ and~$p_1, \ldots, p_n$ be such an instance with~$\sum_{i = 1}^n p$ bounded by a polynomial in~$n$.
    Again we initialize~$G$ with the base gadget shown in Fig.~\ref{fig:treeNPComplete}~(a) and for each task we add one task gadget as shown in Fig.~\ref{fig:treeNPComplete}~(b).
  In the base gadget we set~$d_{max} = \max_{i \in {1, \ldots, n}} d_i$ to be the maximum deadline among all tasks.
  The base and all task gadgets share a common vertex~$u$ at which they are merged.
  The long edge from~$c_1$ to~$c_2$ in the base gadget makes sure that~$u$ is left of all~$a_i$'s and~$b_i$'s on level~$0$.
  Further edges~$e_{i,a}$ and~$e_{i,b}$ force all task gadgets to be nested in each multilevel-planar drawing as in the single-source case in Theorem~\ref{thm:singleSourceNPComplete}.
  Because of this commonality we can adopt its proof to get that a valid schedule for the \textsc{Srtd} exists if and only if a multilevel-planar drawing for tree~$G$ exists.
  \begin{figure}[t]
      \centering
      \includegraphics{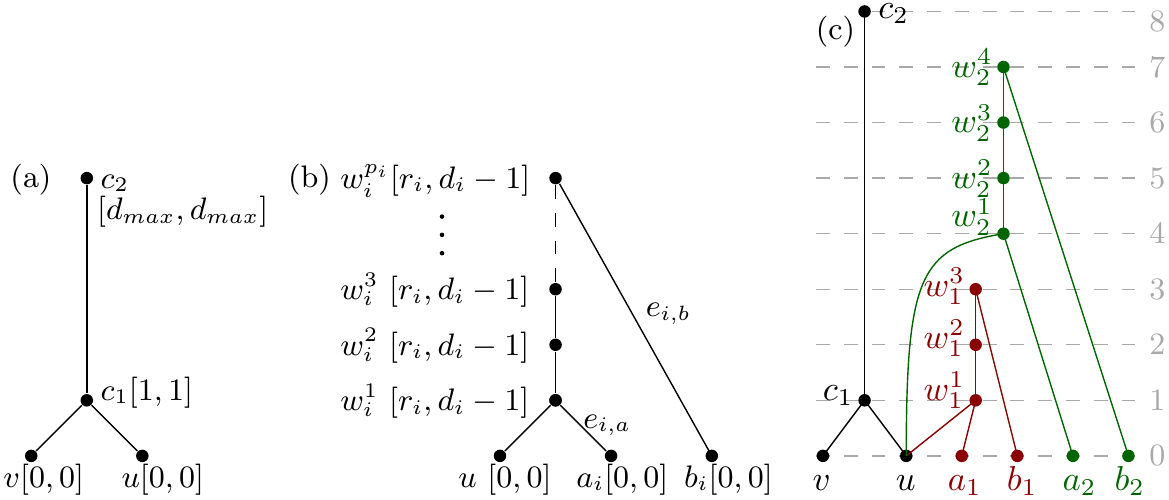}
      \caption{
          The base (a) and the task gadget (b) to transform a \textsc{Srtd} instance into a multilevel-planarity testing instance with~$G$ being a tree. An example (c) with two task gadgets (in red and green).
      }
      \label{fig:treeNPComplete}
  \end{figure}
\end{proof}

This also contrasts the results for upward planarity and level planarity, because every oriented tree is upward planar and all level graphs can be tested for level planarity in linear time.

\subsection{Proof of Theorem~\ref{thm:embeddingNPComplete}}

\begin{proof}
  Suppose that~$\varphi$ is a truth assignment of the variables~$\mathcal V$ that satisfies~$\mathcal C$.
  Construct a drawing~$\Gamma$ of~$G$ that is multilevel planar with respect to~$\ell$ by constructing a level assignment~$\gamma$ as follows.
  Let~$v \in \mathcal V$ be a variable.
  If~$v$ is set to true, set~$\gamma(v) = 1$.
  Otherwise, set~$\gamma(v) = 0$.
  Let~$c_i \in \mathcal C$ be a positive clause.
  Draw the pendulum~$p_i$ of~$c_i$ just below any vertex associated with a variable in~$c_i$ set to true.
  Because~$\varphi$ is a truth assignment that satisfies~$\mathcal C$, such a variable vertex exists.
  Now let~$c_j \in \mathcal C$ be a negative clause.
  Draw the pendulum~$p_j$ of~$c_j$ just above any vertex associated with a variable in~$c_j$ set to false.
  Since~$c_j$ is a negative clause, a positive literal in~$c_j$ corresponds to a variable set to false, and because~$\varphi$ is a truth assignment satisfying~$\mathcal C$, such a variable vertex exists.
  The resulting drawing is then level planar with respect to~$\gamma$ and therefore multilevel planar with respect to~$\ell$.

  Now assume that~$\Gamma$ is a drawing of~$G$ that is multilevel planar with respect to~$\ell$.
  Let~$\gamma$ denote the level assignment induced by~$\Gamma$.
  Construct a truth assignment~$\varphi$ as follows:
  Set the variable~$v \in \mathcal V$ to true or false depending on whether it is~$\gamma(v) = 1$ or~$\gamma(v) = 0$, respectively.
  Because it is~$\ell(v) = \{0, 1\}$, this always assigns a truth value to~$v$.
  Consider the pendulum~$p_i$ of a positive clause~$c_i \in \mathcal C$.
  In a positive gadget,~$p_i$ forces one of the variables in~$c_i$, say~$v$, to level~$1$, i.e.,~$v$ is set to true.
  Because~$c_i$ is a positive clause, it is then satisfied.
  In a negative gadget for a negative clause~$c_j \in \mathcal C$, pendulum~$p_j$ forces one of the variables in~$c_j$, say~$v'$, to level~$0$, i.e.,~$v'$ is set to false.
  Because~$c_j$ is a negative clause, it is then satisfied.
  This means that~$\varphi$ satisfies all clauses in~$\mathcal C$.
\end{proof}

\begin{figure}[t]
  \centering
  \includegraphics{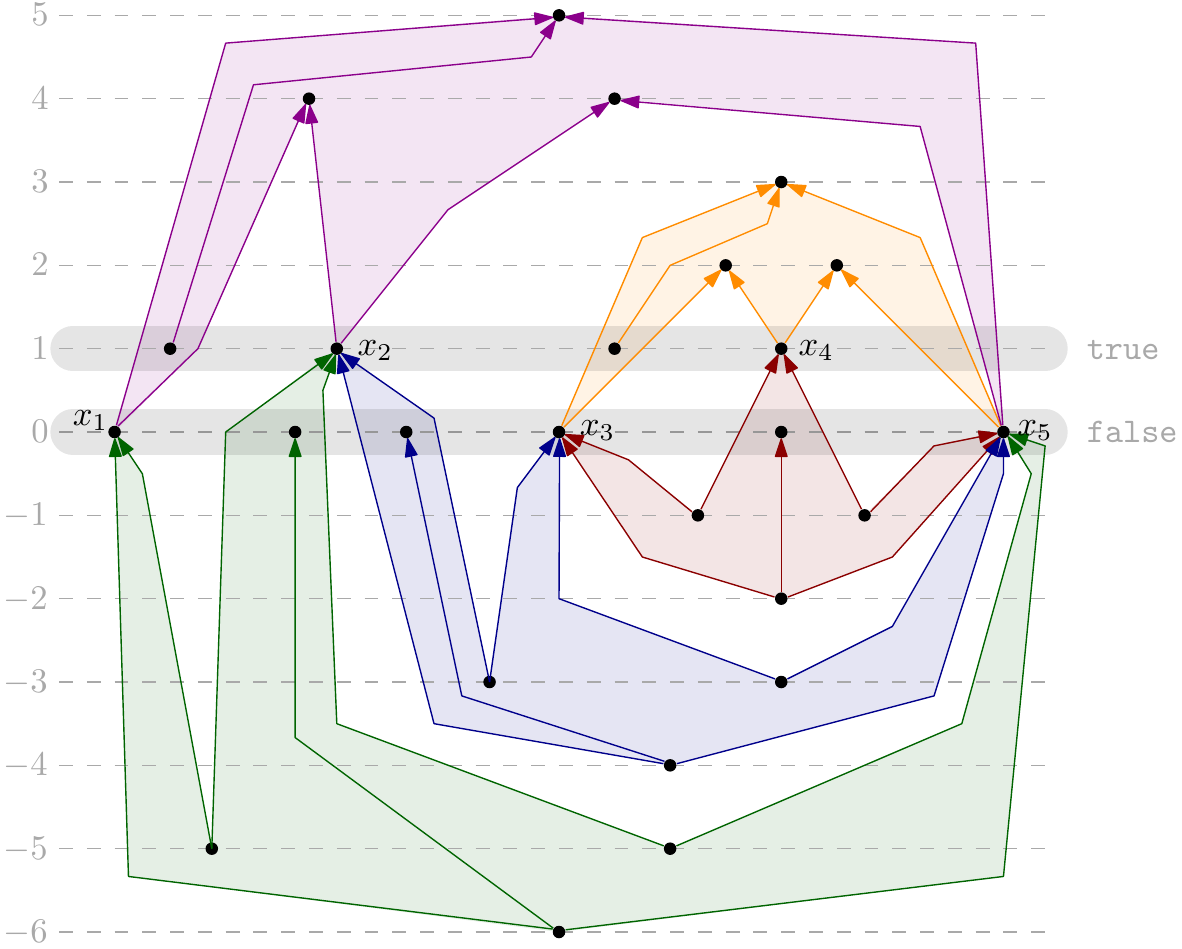}
  \caption{
    A multilevel-planar drawing of the graph constructed from the planar monotone \textsc{3-Sat} instance shown in Fig.~\ref{fig:planarMonotone3SATInstance}.
    The shaded faces correspond to the gadgets that substitute the clauses.
    In this multilevel-planar drawing, vertices~$x_1$,~$x_3$ and~$x_5$ are on level~$0$ so variables~$x_1 = x_3 = x_5 = \text{false}$.
    On the other hand vertices~$x_2$ and~$x_4$ are on level~$1$ so variables~$x_2 = x_4 = \text{true}$.
  }
  \label{fig:embeddingNPCompleteExample}
\end{figure}

\end{document}